%% file: main.tex
\documentclass[10pt, a4paper]{article}
\usepackage{CJKutf8}
\usepackage{geometry}
\input{macro.tex}

\usepackage{authblk}

\begin{document}
\begin{CJK}{UTF8}{gbsn}

\title{Stability analysis of heterogeneous oligopoly games of increasing players with quadratic costs}

\author{Xiaoliang Li\thanks{Corresponding author: xiaoliangbuaa@gmail.com}}
\affil{School of Finance and Trade, Dongguan City College, Dongguan, P. R. China}

%\author[a]{Xiaoliang Li}
%
%\author[b]{Author 2\thanks{Corresponding author: name@mail.com}}
%
%\affil[a]{School of Finance and Trade, Dongguan City College, Dongguan, P. R. China}
%
%\affil[b]{School, University, City, Country}

\maketitle

\begin{abstract}
In this discussion draft, we explore heterogeneous oligopoly games of increasing players with quadratic costs, where the market is supposed to have the isoelastic demand. For each of the models considered in this draft, we analytically investigate the necessary and sufficient condition of the local stability of its positive equilibrium. Furthermore, we rigorously prove that the stability regions are enlarged as the number of involved firms is increasing.
\end{abstract}

%\section{Introduction}

\section{General Assumptions}

Motivated by \cite{Tramontana2015L} , we consider a market served by firms with heterogeneous decision mechanisms producing homogeneous products. We use $q_i(t)$ to denote the output of firm $i$ at period $t$. The cost function of firm $i$ is supposed to be quadratic, i.e., $C_i(q_i)=c q_i^2$. Note that $c$ is a positive parameter and identical for all our firms. Furthermore, assume that the demand function of the market is isoelastic, which is founded on the hypothesis that the consumers have the Cobb-Douglas utility function. Hence, the price of the product should be
$$p(Q)=\frac{1}{Q}=\frac{1}{\sum_i q_i},$$
where $Q=\sum_i q_i$ is the total supply.

\section{Game of Two Firms}\label{sec:duopoly}

%\subsection{Model GB}

First, let us consider a duopoly game, where the first firm adopts a so-called \emph{gradient adjustment mechanism}, while the second firm adopts the \emph{best response mechanism}. Both of these two mechanisms are boundedly rational. To be exact, the first firm increases/decreases its output according to the information given by the marginal profit of the last period, i.e., at period $t+1$,
\begin{equation}
	q_1(t+1)=q_1(t) + k q_1(t) \frac{\partial \Pi_1(t)}{\partial q_1(t)},
\end{equation}
where $\Pi_1(t)=\frac{q_1(t)}{q_1(t)+q_2(t)}-cq_1^2(t)$ is the profit of firm 1 as period $t$, and $k>0$ is a parameter controlling the adjustment speed. It is worth noting that the adjustment speed depends upon not only the parameter $k$ but also the size of the firm $q_1(t)$.

The second firm knows exactly the form of the price function, thus can estimate its profit at period $t+1$ to be 
\begin{equation}
	\Pi_2^e(t+1)=\frac{q_2(t+1)}{q_1^e(t+1)+q_2(t+1)}-cq_2^2(t+1),
\end{equation}
where $q_1^e(t+1)$ is its expectation of the output at period $t+1$ of firm 1. It is realistic that firm 2 has no idea about its rival's production plan of the present period. We suppose that firm 2 have a naive expectation of its competitor to produce the same quantity as the last period, i.e., $q_1^e(t+1)=q_1(t)$. Hence,
\begin{equation}
	\Pi_2^e(t+1)=\frac{q_2(t+1)}{q_1(t)+q_2(t+1)}-cq_2^2(t+1).
\end{equation}
In order to maximize the expected profit, the second firm try to solve the first condition $\partial \Pi_2^e(t+1) / \partial q_2(t+1)=0$, i.e.,
\begin{equation}\label{eq:gb-first-order}
	q_1(t)-2\,cq_2(t+1)(q_1(t)+q_2(t+1))^2=0.
\end{equation}
It should be noted that \eqref{eq:gb-first-order} is an equation of cubic polynomial. Although a general cubic polynomial has at most three real roots, it is easy to know that there exist one single real solution of  \eqref{eq:gb-first-order} for $q_2(t+1)$, but its closed-form expression is particularly complex. However, we suppose that firm 2, by observing the rival's output at the last period, has such ability of computation to find the best response, which is denoted as $R_2(q_1(t))$.

Therefore, the model could be described as the following discrete dynamic system.

\begin{equation}\label{eq:gb-map}
	T_{GB}(q_1,q_2): 
	\left\{\begin{split}
		&q_1(t+1)=q_1(t) + k q_1(t)\left[\frac{q_2(t)}{(q_1(t)+q_2(t))^2}-2\,cq_1(t)\right],\\
		&q_2(t+1)=R_2(q_1(t)).
	\end{split}
	\right.
\end{equation}

By setting $q_1(t+1)=q_1(t)=q_1$ and $q_2(t+1)=q_2(t)=q_2$, the equilibrium can be identified by
\begin{equation}
	\left\{\begin{split}
		&q_1=q_1 + k q_1\left(\frac{q_2}{(q_1+q_2)^2}-2\,cq_1\right),\\
		&q_2=R_2(q_1),
	\end{split}
	\right.
\end{equation}
where $q_2=R_2(q_1)$ can be reformulated to $q_1-2\,cq_2(q_1+q_2)^2=0$ according to \eqref{eq:gb-first-order}. Thus, we have
\begin{equation}
	\left\{\begin{split}
		&k q_1\left(\frac{q_2}{(q_1+q_2)^2}-2\,cq_1\right)=0,\\
		&q_1-2\,cq_2(q_1+q_2)^2=0,
	\end{split}
	\right.
\end{equation}
which could be solved by a unique solution
$$E_{GB}^1=\left(\frac{1}{\sqrt{8c}},\frac{1}{\sqrt{8c}}\right).$$ 
It should be noted that $(0,0)$ is not an equilibrium for it is not defined for the iteration map \eqref{eq:gb-map}. In order to investigate the local stability of an equilibrium $(q_1^*,q_2^*)$, we consider the Jacobian matrix of the form
\begin{equation}
	J_{GB}(q_1^*,q_2^*)=\left[\begin{matrix}
		\frac{\partial q_1(t+1)}{\partial q_1(t)}\big|_{(q_1^*,q_2^*)} & \frac{\partial q_1(t+1)}{\partial q_2(t)}\big|_{(q_1^*,q_2^*)}\\
		\frac{\partial q_2(t+1)}{\partial q_1(t)}\big|_{(q_1^*,q_2^*)} & \frac{\partial q_2(t+1)}{\partial q_2(t)}\big|_{(q_1^*,q_2^*)}\\
	\end{matrix}
	\right].
\end{equation}

It is easy to obtain that
\begin{equation}
	\begin{split}
& \frac{\partial q_1(t+1)}{\partial q_1(t)}\big|_{(q_1^*,q_2^*)}= 1+kq_2^*\frac{q_2^*-q_1^*}{(q_1^*+q_2^*)^3}-4\,ckq_1^*,\\
&\frac{\partial q_1(t+1)}{\partial q_2(t)}\Big|_{(q_1^*,q_2^*)}= kq_1^*\frac{q_1^*-q_2^*}{(q_1^*+q_2^*)^3}.
	\end{split}
\end{equation}
Furthermore, the derivative of $q_2(t+1)$ with respect to $q_2(t)$ is $0$ as $R_2$ does not involve $q_2$. However, the derivative of $q_2(t+1)$  with respect to $q_1(t)$ may not be directly obtained. By virtue of the method called implicit differentiation, it can be acquired that
\begin{equation}
	\frac{\partial q_2(t+1)}{\partial q_1(t)}\Big|_{(q_1^*,q_2^*)}
	=-\frac{4\,cq_1^*q_2^*+4\,cq_2^{*2}-1}{2\,c(q_1^{*2}+4\,q_1^*q_2^*+3\,q_2^{*2})}.
\end{equation}

At $E_{GB}^1=(1/\sqrt{8c},1/\sqrt{8c})$, we have that
\begin{equation}
	J_{GB}(E_{GB}^1)=\left[\begin{matrix}
		1-k\sqrt{2\,c} & 0\\
		0 & 0\\
	\end{matrix}
	\right].
\end{equation}
Obviously, its eigenvalues are $\lambda_1=1-k\sqrt{2\,c}$ and $\lambda_2=0$. It is evident that $E_{GB}^1$ is locally stable if and only if $k\sqrt{c}<\sqrt{2}$. We summarize the above results in the following proposition.

\begin{proposition}
	The $T_{GB}$ model described by \eqref{eq:gb-map} has a unique  equilibrium 
	$$\left(\frac{1}{\sqrt{8c}}, \frac{1}{\sqrt{8c}}\right),$$ 
	which is locally stable provided that
\begin{equation}\label{eq:gb-stable-cd}
	k\sqrt{c}<\sqrt{2}.
\end{equation}
\end{proposition}

\section{Game of Three Firms}

%\subsection{Model GBA}

In this section, we introduce a new boundedly rational player and add it to the model of the previous section. This player is assumed to take an \emph{adaptive mechanism}, which means that at each period $t+1$ it decides the quantity of production $q_3(t+1)$ according to the previous output $q_3(t)$ as well as its expectations of the other two competitors. It is also supposed that this player naively expects that at period $t+1$ firm 1 and 2 would produce the same quantity as at period $t$. Therefore, the third firm could calculate the best response $R_3(q_1(t),q_2(t))$ to maximize its expected profit. Similar as \eqref{eq:gb-first-order}, $R_3(q_1(t),q_2(t))$ is the solution for $q_3'(t+1)$ of the following equation.
\begin{equation}
	q_1(t)+q_2(t)-2\,cq_3'(t+1)(q_1(t)+q_2(t)+q_3'(t+1))^2=0.
\end{equation}
The adaptive decision mechanism for firm 3 is that it choose the output $q_3(t+1)$ proportionally to be
$$q_3(t+1)=(1-l)q_3(t)+lR_3(q_1(t),q_2(t)),$$
where $l\in(0,1]$ is a parameter controlling the proportion. 

Hence, the triopoly can be described by

\begin{equation}\label{eq:gba-map}
	T_{GBA}(q_1,q_2,q_3): 
	\left\{\begin{split}
		&q_1(t+1)=q_1(t) + k q_1(t)\left[\frac{q_2(t)+q_3(t)}{(q_1(t)+q_2(t)+q_3(t))^2}-2\,cq_1(t)\right],\\
		&q_2(t+1)=R_2(q_1(t),q_3(t)),\\
		&q_3(t+1)=(1-l)q_3+l R_3(q_1(t),q_2(t)).
		\end{split}
	\right.
\end{equation}
Similar to Section \ref{sec:duopoly}, the equilibria satisfy that
\begin{equation}
	\left\{\begin{split}
		&k q_1\left(\frac{q_2+q_3}{(q_1+q_2+q_3)^2}-2\,cq_1\right)=0,\\
		&q_1+q_3-2\,cq_2(q_1+q_2+q_3)^2=0,\\
		&q_1+q_2-2\,cq_3(q_1+q_2+q_3)^2=0,
	\end{split}
	\right.
\end{equation}
which could be solved by
\begin{equation}
\begin{split}
E_{GBA}^1=&~\left(0,\frac{1}{\sqrt{8c}},\frac{1}{\sqrt{8c}}\right),\\
E_{GBA}^2=&~\left(\frac{1}{\sqrt{9c}},\frac{1}{\sqrt{9c}},\frac{1}{\sqrt{9c}}\right).	
\end{split}
\end{equation}

For an equilibrium $(q_1^*,q_2^*,q_3^*)$, the Jacobian matrix of $T_{GBA}$ takes the form
\begin{equation}
	J_{GBA}(q_1^*,q_2^*,q_3^*)=\left[\begin{matrix}
		\frac{\partial q_1(t+1)}{\partial q_1(t)}\big|_{(q_1^*,q_2^*,q_3^*)} & \frac{\partial q_1(t+1)}{\partial q_2(t)}\big|_{(q_1^*,q_2^*,q_3^*)} & \frac{\partial q_1(t+1)}{\partial q_3(t)}\big|_{(q_1^*,q_2^*,q_3^*)}\\
		\frac{\partial q_2(t+1)}{\partial q_1(t)}\big|_{(q_1^*,q_2^*,q_3^*)} & \frac{\partial q_2(t+1)}{\partial q_2(t)}\big|_{(q_1^*,q_2^*,q_3^*)} & \frac{\partial q_2(t+1)}{\partial q_3(t)}\big|_{(q_1^*,q_2^*,q_3^*)}\\
		\frac{\partial q_3(t+1)}{\partial q_1(t)}\big|_{(q_1^*,q_2^*,q_3^*)} & \frac{\partial q_3(t+1)}{\partial q_2(t)}\big|_{(q_1^*,q_2^*,q_3^*)} & \frac{\partial q_3(t+1)}{\partial q_3(t)}\big|_{(q_1^*,q_2^*,q_3^*)}\\
	\end{matrix}
	\right].
\end{equation}
The first and the second rows of the matrix might be similarly computed as Section \ref{sec:duopoly}. For the third row, we have
\begin{equation}
\begin{split}
	\frac{\partial q_3(t+1)}{\partial q_1(t)}=&~l\frac{\partial R_3(q_1(t),q_2(t))}{\partial q_1(t)},\\
	\frac{\partial q_3(t+1)}{\partial q_2(t)}=&~l\frac{\partial R_3(q_1(t),q_2(t))}{\partial q_2(t)},\\
	\frac{\partial q_3(t+1)}{\partial q_3(t)}=&~1-l,\\
\end{split}
\end{equation}
where ${\partial R_3(q_1(t),q_2(t))}/{\partial q_1(t)}$ and ${\partial R_3(q_1(t),q_2(t))}/{\partial q_2(t)}$ can be acquired using the method of implicit differentiation. From an economic point of view, we only consider the positive equilibrium $E^2_{GBA}$, where the Jacobian matrix would be
\begin{equation}
	J_{GBA}(E^2_{GBA})=\left[\begin{matrix}
		1-{10\,k\sqrt{c}}/{9} & -k\sqrt{c}/9 & -k\sqrt{c}/9\\
		-{1}/{10} & 0 & -{1}/{10}\\
		-{l}/{10} & -{l}/{10} & 1-l
	\end{matrix}
	\right].
\end{equation}

Let $A$ be the characteristic polynomial of a Jacobian matrix  $J$. The eigenvalues of $J$ are simply the roots of the polynomial $A$ for $\lambda$. So the problem of stability analysis can be reduced to that of determining whether all the roots of $A$ lie in the open unit disk $|\lambda|<1$.  To the best of our knowledge, in addition to the Routh-Hurwitz criterion \cite{Oldenbourg1948T} generalized from the corresponding criterion for continuous systems, there are two other criteria, the Schur-Cohn criterion \cite[pp.\,246--248]{Elaydi2005U} and the Jury criterion \cite{Jury1976I}, available for discrete dynamical systems. In what follows, we provide a short review of the Schur-Cohn criterion.

\begin{proposition}[Schur-Cohn Criterion]\label{prop:Schur-Cohn}
For a $n$-dimensional discrete dynamic system, assume that the characteristic polynomial of its Jacobian matrix is
\begin{equation*}
A= \lambda^n + a_{n-1}\lambda^{n - 1} + \cdots + a_0.
\end{equation*}
Consider the
sequence of determinants $D^\pm_1$, $D^\pm_2$, $\ldots$, $D^\pm_n$,
where
\begin{equation*}
\begin{split}
  D^{\pm}_i =&\left| \left(
\begin{array}{ccccc}
1&a_{n-1}&a_{n-2}&\cdots&a_{n-i+1}\\
0&1&a_{n-1}&\cdots&a_{n-i+2}\\
0&0&1&\cdots&a_{n-i+3}\\
\vdots&\vdots&\vdots&\ddots&\vdots\\
0&0&0&\cdots&1\\
\end{array}
\right)\pm\left(
\begin{array}{ccccc}
a_{i-1}&a_{i-2}&\cdots&a_{1}&a_0\\
a_{i-2}&a_{i-3}&\cdots&a_{0}&0\\
\vdots&\vdots&\ddots&\vdots&\vdots\\
a_{1}&a_0&\cdots&0&0\\
a_0&0&\cdots&0&0\\
\end{array}
\right)\right|.
\end{split}
\end{equation*}
The characteristic polynomial $A$ has all its roots inside the unit
open disk if and only if
\smallskip
\begin{enumerate}
\item $A(1)>0$ and $(-1)^nA(-1)>0$,

\item  $D^\pm_1>0, D^\pm_3>0, \ldots, D^\pm_{n-3}>0, D^\pm_{n-1}>0$
(when $n$ is even), or\\[2pt]
\smallskip $D^\pm_2>0, D^\pm_4>0, \ldots, D^\pm_{n-3}>0,
D^\pm_{n-1}>0$ (when $n$ is odd).
\end{enumerate}
\end{proposition}

\begin{corollary}
	Consider a $3$-dimensional discrete dynamic system with the characteristic polynomial of its Jacobian matrix of the form 
	$$A=\lambda^3+a_2\lambda^2+a_1\lambda+a_0.$$
	An equilibrium $E$ is locally stable if and only if the following inequalities are satisfied at $E$.
	\begin{equation}\label{eq:cd4-stable}
		\left\{\begin{split}
			&1+a_2+a_1+a_0>0,\\
			&1-a_2+a_1-a_0>0,\\
			&-a_0^2-a_0a_2+a_1+1>0,\\
			&-a_0^2+a_0a_2-a_1+1>0.
		\end{split}\right.
	\end{equation}
\end{corollary}

For the $3$-dimensional discrete dynamic system \eqref{eq:gba-map}, it is easy to verify that at the unique positive equilibrium $E_{GBA}^2$ the local stability condition \eqref{eq:cd4-stable}  could be reformulated to 
\begin{equation}\label{eq:inequality-4}
	CD_{GBA}^1>0,~CD_{GBA}^2>0,~CD_{GBA}^3<0,~CD_{GBA}^4<0,
\end{equation}
where
\begin{equation}
\begin{split}
CD_{GBA}^1=&~
kl\sqrt{c},\\	
CD_{GBA}^2=&~
504\,kl\sqrt{c}-1010\,k\sqrt{c}-909\,l+1800,\\
CD_{GBA}^3=&~
324\,ck^2l^2-18360\,ck^2l+10100\,ck^2-16524\,kl^2\sqrt{c}-840420\,kl\sqrt{c}\\
&+8181\,l^2+891000\,k\sqrt{c}+801900\,l-1620000,\\
CD_{GBA}^4=&~
36\,ck^2l^2+1960\,ck^2l+1764\,kl^2\sqrt{c}-1100\,ck^2+93420\,kl\sqrt{c}\\
&-99000\,k\sqrt{c}-891\,l^2-89100\,l.
\end{split}
\end{equation}

It is obvious that $CD_{GBA}^1>0$ could be ignored as it is always true for all parameter values such that $k>0$, $c>0$ and $1\geq l>0$. A further question is whether the other three inequalities could be simplified. To answer this question, we might investigate the inclusion relations of these inequalities. It  worth noticing that the surfaces $CD_{GBA}^2=0$, $CD_{GBA}^3=0$ and $CD_{GBA}^4=0$ divide the parameter space $\{(k,l,c)\,|\,k>0,1\geq l>0,c>0\}$ of our concern into a number of connected regions. Moreover, in a given region, the signs of $CD_{GBA}^i$ ($i=1,2,3,4$) would be invariant. This means that in each of these regions we could identify whether the inequalities in \eqref{eq:inequality-4} are satisfied by checking them at a single sample point. For simple cases, the selection of sample points might be done by hand. Generally, however, the selection could be automated by using, e.g., the partial cylindrical algebraic decomposition (PCAD) method \cite{Collins1991P}. 

\begin{table}[htbp]
	\centering 
	\caption{Stability Condition of $T_{GBA}$ at Selected Sample Points}
	\label{tab:gba-sample} 
	 
	\begin{tabular}{|l|c|c|c|c|}  
		\hline  % 表格的横线
		 sample point of $(k,l,c)$ & $CD_{GBA}^1>0$ & $CD_{GBA}^2>0$ & $CD_{GBA}^3<0$ & $CD_{GBA}^4<0$\\ 
		\hline
		(455/256, 71/256, 1/4) & true & true & true & true \\
		\hline
		(31/8, 71/256, 1/4) & true & false & true & true \\
		\hline
		(601/128, 71/256, 1/4)&true& false& false& true\\
		\hline
		(453/256, 183/256, 1/4)&true & true& true& true\\
		\hline
		(1439/256, 183/256, 1/4) &true& false& true& true\\
		\hline
		(1577/16, 183/256, 1/4)&true& false& false& true\\
		\hline
		(49855/256, 183/256, 1/4)&true& false& true& true \\
		\hline
		(25673/128, 183/256, 1/4)& true& false& true& false\\
		\hline
		(451/256, 15/16, 1/4)&true& true& true& true\\
		\hline
		(5237/256, 15/16, 1/4)&true& false& true& true\\
		\hline
		(2425/64, 15/16, 1/4)&true& false& true& false\\
		\hline
	\end{tabular}
\end{table}

In Table \ref{tab:gba-sample}, we list all the selected sample points such that there is at least one point in each region. The four inequalities in \eqref{eq:inequality-4} are verified at these sample points one by one, which are also given in Table \ref{tab:gba-sample}. It is observed that at the sample points where  $CD_{GBA}^2>0$ is true, the other three inequalities would also be true. Hence, if $CD_{GBA}^2>0$ is satisfied, then all the four inequalities in \eqref{eq:inequality-4} would be satisfied definitely. In other words, only $CD_{GBA}^2>0$ is needed herein for the detection of the local stability. Furthermore, $CD_{GBA}^2>0$ is equivalent to
\begin{equation*}
	k\sqrt{c}<\frac{9(101\,l-200)}{2(252\,l-505)}.
\end{equation*}
Therefore, we summarize the obtained results in the following proposition.

\begin{proposition}
The $T_{GBA}$ model described by \eqref{eq:gba-map} has a unique positive equilibrium 
$$\left(\frac{1}{\sqrt{9c}}, \frac{1}{\sqrt{9c}}, \frac{1}{\sqrt{9c}}\right),$$
 which is locally stable provided that
\begin{equation}
	k\sqrt{c}<\frac{9(101\,l-200)}{2(252\,l-505)}.
\end{equation}
%\begin{equation}
%ck^2<\frac{826281\,l^2-3272400\,l+3240000}{254016\,l^2-1018080\,l+1020100}
%\end{equation}
\end{proposition}

Furthermore, we have the following result.

\begin{proposition}
	The stability region of the $T_{GBA}$ model is strictly larger than that of $T_{GB}$.
\end{proposition}
\begin{proof}
It suffices to prove that 
	$$\frac{9(101\,l-200)}{2(252\,l-505)}>\sqrt{2},$$
which is equivalent to 
\begin{equation*}\label{eq:gba-gb}
	9(101\,l-200) < 2\sqrt{2}(252\,l-505)
\end{equation*}
since $252\,l-505<0$. It is easy to see that the above inequality can be reformulated to
$$(909-504\sqrt{2})l<(1800-1010\sqrt{2}),$$
which is true by checking at $l=0$ and $l=1$. This completes the proof.
\end{proof}

%\begin{proof}
%It would suffice to prove that 
%	$$\frac{826281\,l^2-3272400\,l+3240000}{254016\,l^2-1018080\,l+1020100}
%	=2+\frac{318249 \,l^2  - 1236240\, l + 1199800}{254016\,l^2-1018080\,l+1020100}>2.$$
%It is noted that if $0<l\leq 1$, then
%$$318249 \,l^2  - 1236240\, l + 1199800>0$$ 
%and 
%$$254016\,l^2-1018080\,l+1020100>0$$ 
%for each of the left parts has two real zeros greater than $1$. This completes the proof.
%\end{proof}

\section{Game of Four Firms}

%\subsection{Model GBAL}

In this section, we introduce an additional player. The fourth firm adopts the so-called \emph{local monopolistic approximation} (LMA) mechanism \cite{Tuinstra2004A}, which is also a boundedly rational adjustment process. In this process, the player just has limited knowledge of the demand function. To be exact, the firm can observe the current market price $p(t)$ and the corresponding total supply $Q(t)$ and is able to correctly estimate the slope $p'(Q(t))$ of the price function around the point $(p(t),Q(t))$. Then, the firm uses such information to conjecture the demand function and expect the price at period $t+1$ to be
$$p^e(t+1)=p(Q(t))+p'(Q(t))(Q^e(t+1)-Q(t)),$$
where $Q^e(t+1)$ represents the expected aggregate production at period $t+1$. Moreover, firm $4$ is also assumed to use the naive expectations of its rivals, i.e., 
$$Q^e(t+1)=q_1(t)+q_2(t)+q_3(t)+q_4(t+1).$$ 
Thus, we have that 
$$p^e(t+1)=\frac{1}{Q(t)}-\frac{1}{Q^2(t)}(q_4(t+1)-q_4(t)).$$
The expected profit of the fourth firm is
$$\Pi^e_4(t+1)=p^e(t+1)q_4(t+1)-cq_4^2(t+1).$$
To maximize the expected profit, firm $4$ chooses its output at period $t+1$ to be the solution of the first order condition
$$q_4(t+1)=\frac{2\,q_4(t)+q_1(t)+q_2(t)+q_3(t)}{2(1+c(q_1(t)+q_2(t)+q_3(t)+q_4(t))^2)}.$$

Therefore, the new model can be described by the following $4$-dimensional discrete dynamic system.
\begin{equation}\label{eq:gbal-map}
\begin{split}
	&T_{GBAL}(q_1,q_2,q_3,q_4): \\
	&\left\{\begin{split}
		&q_1(t+1)=q_1(t) + k q_1(t)\left[\frac{q_2(t)+q_3(t)+q_4(t)}{(q_1(t)+q_2(t)+q_3(t)+q_4(t))^2}-2\,cq_1(t)\right],\\
		&q_2(t+1)=R_2(q_1(t),q_3(t),q_4(t)),\\
		&q_3(t+1)=(1-l)q_3+l R_3(q_1(t),q_2(t),q_4(t)),\\
		&q_4(t+1)=\frac{2\,q_4(t)+q_1(t)+q_2(t)+q_3(t)}{2(1+c(q_1(t)+q_2(t)+q_3(t)+q_4(t))^2)}.
		\end{split}
	\right.
\end{split}
\end{equation}
Similarly, we know that the equilibria are described by
\begin{equation}
	\left\{\begin{split}
		&k q_1\left(\frac{q_2+q_3+q_4}{(q_1+q_2+q_3+q_4)^2}-2\,cq_1\right)=0,\\
		&q_1+q_3+q_4-2\,cq_2(q_1+q_2+q_3+q_4)^2=0,\\
		&q_1+q_2+q_4-2\,cq_3(q_1+q_2+q_3+q_4)^2=0,\\
		&q_4-\frac{2\,q_4+q_1+q_2+q_3}{2(1+c(q_1+q_2+q_3+q_4)^2)}=0,
	\end{split}
	\right.
\end{equation}
which could be solved by two solutions 
\begin{equation}
\begin{split}
E_{GBAL}^1=&~\left(0,\frac{1}{\sqrt{9c}},\frac{1}{\sqrt{9c}},\frac{1}{\sqrt{9c}}\right),\\	
E_{GBAL}^2=&~\left(\sqrt{\frac{3}{32c}},\sqrt{\frac{3}{32c}},\sqrt{\frac{3}{32c}},\sqrt{\frac{3}{32c}}\right).\\	
\end{split}
\end{equation}
Hence, there exists a unique positive equilibrium $E_{GBAL}^2$, where the Jacobian matrix of $T_{GBAL}$ should be
\begin{equation}
	J_{GBAL}(E^2_{GBAL})=\left[\begin{matrix}
		1-{3\,k\sqrt{6c}}/8 & -k\sqrt{6c}/{24} & -k\sqrt{6c}/{24} & -k\sqrt{6c}/{24}\\
		-{1}/{9} & 0 & -{1}/{9} & -{1}/{9}\\
		-{l}/{9} & -{l}/{9} & 1-l & -{l}/{9}\\
		-{1}/{10} & -{1}/{10} & -{1}/{10} & {1}/{10}\\
	\end{matrix}
	\right].
\end{equation}

By virtue of Proposition \ref{prop:Schur-Cohn}, we have the following corollary.

\begin{corollary}\label{cor:4-dim-stable}
	Consider a $4$-dimensional discrete dynamic system with the characteristic polynomial of its Jacobian matrix of the form 
	$$A=\lambda^4+a_3\lambda^3+a_2\lambda^2+a_1\lambda+a_0.$$
	An equilibrium $E$ is locally stable if and only if the following inequalities are satisfied at $E$.
	\begin{equation}\label{eq:cd6-stable}
		\left\{\begin{split}
			&1+a_3+a_2+a_1+a_0>0,\\
			&1-a_3+a_2-a_1+a_0>0,\\
			&-a_0^3-a_0^2a_2+a_0a_1a_3+a_0a_3^2-a_0^2-a_1^2-a_1a_3+a_0+a_2+1>0,\\
			&a_0^3-a_0^2a_2+a_0a_1a_3-a_0a_3^2-a_0^2+2\,a_0a_2-a_1^2+a_1a_3-a_0-a_2+1>0,\\
			&1+a_0>0,\\
			&1-a_0>0.
		\end{split}\right.
	\end{equation}
\end{corollary}

For the $4$-dimensional discrete dynamic system \eqref{eq:gbal-map}, it is easy to verify that at the unique positive equilibrium $E_{GBAL}^2$ the above condition \eqref{eq:cd6-stable} could be reformulated to 
\begin{equation}\label{eq:inequality-6}
\begin{split}
	CD_{GBAL}^1>0,~CD_{GBAL}^2>0,~CD_{GBAL}^3>0,\\
	CD_{GBAL}^4<0,~CD_{GBAL}^5<0,~CD_{GBAL}^6>0,
\end{split}
\end{equation}
where
\begin{equation}
\begin{split}
CD_{GBAL}^1=&~
kl\sqrt{32c/3},\\	
CD_{GBAL}^2=&~
(512\, k l - 1017\, k )\sqrt{32c/3} - 3616\, l + 7056,\\
CD_{GBAL}^3=&~
(28672\,k^3l^3-1062432\,k^3l^2+9180054\,k^3l-12603681\,k^3)(\sqrt{32c/3})^3\\
&+(-3777536\,k^2l^3+179157888\,k^2l^2-1194862752\,k^2l+945483840\,k^2)(\sqrt{32c/3})^2\\
&+(116054016\,kl^3-4248400896\,kl^2-5573546496\,kl+13237426944\,k)\sqrt{32c/3}\\
&-566525952\,l^3+11952783360\,l^2+47066406912\,l-133145026560,\\
CD_{GBAL}^4=&~
(3616\,k^3l^3-132966\,k^3l^2-512973\,k^3l+1226907\,k^3)(\sqrt{32c/3})^3\\
&+(-472768\,k^2l^3+16419744\,k^2l^2+77813136\,k^2l-83525904\,k^2)(\sqrt{32c/3})^2\\
&+(-6484992\,kl^3+276668928\,kl^2+1145829888\,kl-1868106240\,k)\sqrt{32c/3}\\
&+55148544\,l^3-1055932416\,l^2-6642155520\,l,\\
CD_{GBAL}^5=&~
(16\,kl-27\,k)\sqrt{32c/3}-96\,l-12816,\\
CD_{GBAL}^6=&~
(16\,kl-27\,k)\sqrt{32c/3}-96\,l+13104,\\
\end{split}
\end{equation}

%\begin{table}[htbp]
%	\centering 
%	\caption{Stability Condition of $T_{GBAL}$ at Selected Sample Points}
%	\label{tab:gba-sample} 
%	 
%	\begin{tabular}{|l|c|c|c|c|c|c|}  
%		\hline  % 表格的横线
%		 sample point of $(k,l,c)$ & $CD_{GBAL}^1>0$ & $CD_{GBA}^2>0$ & $CD_{GBA}^3<0$ & $CD_{GBA}^4<0$ & $CD_{GBAL}^1>0$ & $CD_{GBAL}^1>0$\\ 
%		\hline 
%		(55/64, 109/256, 3/2) & true& true& true& true& true& true \\
%		\hline 
%		(243/128, 109/256, 3/2) & true& false& true& true& true& true  \\
%		\hline
%		(301/32, 109/256, 3/2)& true& false& false& true& true& true\\
%		\hline
%		(271/16, 109/256, 3/2)& true& false& true& true& true& true\\
%		\hline
%		(5725/64, 109/256, 3/2)& true& false& true& false& true& true\\
%		\hline
%		(20771/128, 109/256, 3/2)& true& false& true& false& true& false\\
%		\hline
%		(109/128, 119/128, 3/2)& true& true& true& true& true& true \\
%		\hline
%		(1275/256, 119/128, 3/2)& true& false& true& true& true& true\\
%		\hline
%		(35405/256, 119/128, 3/2)& true& false& true& false& true& true\\
%		\hline
%		(34413/128, 119/128, 3/2)& true& false& true& false& true& false\\
%		\hline
%	\end{tabular}
%\end{table}

\begin{table}[htbp]
	\centering 
	\caption{Stability Condition of $T_{GBAL}$ at Selected Sample Points}
	\label{tab:gbal-sample} 
	 
	\begin{tabular}{|l|c|c|c|}  
		\hline  % 表格的横线
		 sample point of $(k,l,c)$ & $CD_{GBAL}^1>0$ & $CD_{GBAL}^2>0$ & $CD_{GBAL}^3>0$ \\ 
		\hline 
		(55/64, 109/256, 3/2) & true& true& true \\
		\hline 
		(243/128, 109/256, 3/2) & true& false& true  \\
		\hline
		(301/32, 109/256, 3/2)& true& false& false\\
		\hline
		(271/16, 109/256, 3/2)& true& false& true\\
		\hline
		(5725/64, 109/256, 3/2)& true& false& true\\
		\hline
		(20771/128, 109/256, 3/2)& true& false& true\\
		\hline
		(109/128, 119/128, 3/2)& true& true& true\\
		\hline
		(1275/256, 119/128, 3/2)& true& false& true\\
		\hline
		(35405/256, 119/128, 3/2)& true& false& true\\
		\hline
		(34413/128, 119/128, 3/2)& true& false& true\\
		\hline
	\end{tabular}
	
		\begin{tabular}{|l|c|c|c|}  
		\hline  % 表格的横线
		 sample point of $(k,l,c)$ & $CD_{GBAL}^4<0$ & $CD_{GBAL}^5<0$ & $CD_{GBAL}^6>0$\\ 
		\hline 
		(55/64, 109/256, 3/2) & true& true& true \\
		\hline 
		(243/128, 109/256, 3/2) &  true& true& true  \\
		\hline
		(301/32, 109/256, 3/2)&  true& true& true\\
		\hline
		(271/16, 109/256, 3/2)&  true& true& true\\
		\hline
		(5725/64, 109/256, 3/2)& false& true& true\\
		\hline
		(20771/128, 109/256, 3/2)&  false& true& false\\
		\hline
		(109/128, 119/128, 3/2)&  true& true& true \\
		\hline
		(1275/256, 119/128, 3/2)&  true& true& true\\
		\hline
		(35405/256, 119/128, 3/2)&  false& true& true\\
		\hline
		(34413/128, 119/128, 3/2)&  false& true& false\\
		\hline
	\end{tabular}

\end{table}

In order to simplify condition \eqref{eq:inequality-6}, it is also helpful to explore the inclusion relations of these inequalities. Bear in mind that the surfaces $CD_{GBAL}^i=0$ ($i=1,\ldots,6$) divide the parameter space $\{(k,l,c)\,|\,k>0,1\geq l>0,c>0\}$ into regions, and in each of them the signs of $CD_{GBA}^i$ ($i=1,\ldots,6$) would be invariant. Similarly, we use the PCAD method to select at least one sample point from each region.

Table \ref{tab:gbal-sample} lists the selected sample points and shows the verification results of the six inequalities in \eqref{eq:inequality-6} at these sample points. It is observed that at all the sample points where  $CD_{GBAL}^2>0$ is true, the rest inequalities would also be true. In other words, if $CD_{GBAL}^2>0$, then the local stability condition \eqref{eq:cd6-stable} would be satisfied. Thus, condition \eqref{eq:inequality-6} could be simplified to one single inequality. Furthermore, it is easy to see that $CD_{GBAL}^2>0$ is equivalent to
\begin{equation*}
k\sqrt{c}<\frac{2\sqrt{6}(226\,l-441)}{512\,l-1017}.
\end{equation*}
Therefore, we summarize the results in the following proposition.

\begin{proposition}
The $T_{GBAL}$ model described by \eqref{eq:gbal-map} has a unique positive equilibrium 
$$\left(\sqrt{\frac{3}{32c}},\sqrt{\frac{3}{32c}},\sqrt{\frac{3}{32c}},\sqrt{\frac{3}{32c}}\right),$$
 which is locally stable provided that
\begin{equation*}
k\sqrt{c}<\frac{2\sqrt{6}(226\,l-441)}{512\,l-1017}.
\end{equation*}
\end{proposition}

Furthermore, we have the following result.

\begin{proposition}
	The stability region of the $T_{GBAL}$ model is strictly larger than that of $T_{GBA}$.
\end{proposition}
\begin{proof}
It suffices to prove that 
	$$\frac{9(101\,l-200)}{2(252\,l-505)}<\frac{2\sqrt{6}(226\,l-441)}{512\,l-1017},$$
which is equivalent to 
$$9(101\,l-200)(512\,l-1017) < 4\sqrt{6}(252\,l-505)(226\,l-441),$$
and further to
$$ (-227808\sqrt{ 6} + 465408)l^2 + (901048 \sqrt{6} - 1846053)l - 890820\sqrt{6} + 1830600<0.$$
This inequality is satisfied for $0<l\leq 1$ since the left part has a negative leading coefficient and has both of its roots greater than $1$, which completes the proof.
\end{proof}

\section{Game of Five Firms}

Finally, we introduce a special firm, which is a {rational player}, to the model of this section. A \emph{rational player}, quite different from the second player, not only knows clearly the form of the price function, but also has complete information of its rivals' decisions. Because of no information about the rivals, firm 2 just naively  expects that all its competitors produce the same amounts as the last period. Thus, the expected profit of firm 2 at period $t+1$ would be
$$\Pi_2^e(t+1)=\frac{q_2(t+1)}{q_1(t)+q_2(t+1)+q_3(t)+q_4(t)+q_5(t)}-cq_2^2(t+1).$$
In comparison, firm 5 has complete information and know exactly the production plans of all its rivals. Hence, the expected profit of firm 5 would be the real profit, i.e.,
$$\Pi_5^e(t+1)=\Pi_5(t+1)=\frac{q_5(t+1)}{q_1(t+1)+q_2(t+1)+q_3(t+1)+q_4(t+1)+q_5(t+1)}-cq_5^2(t+1).$$
In order to maximize its profit, firm 5 need to solve the first condition $\partial \Pi_5(t+1) / \partial q_5(t+1)=0$ for $q_5(t+1)$. We denote the solution as
$$q_5(t+1)=R_5(q_1(t+1),q_2(t+1),q_3(t+1),q_4(t+1)).$$
It is worth noting that the form of the solution is similar as that of firm 2, but with variables replaced by the output quantities of the rivals at the present period. In short, we have the $5$-dimensional iteration map
\begin{equation}\label{eq:gbalr-map}
\begin{split}
	&T_{GBALR}(q_1,q_2,q_3,q_4,q_5): \\
	&\left\{\begin{split}
		&q_1(t+1)=q_1(t) + k q_1(t)\left[\frac{q_2(t)+q_3(t)+q_4(t)+q_5(t)}{(q_1(t)+q_2(t)+q_3(t)+q_4(t)+q_5(t))^2}-2\,cq_1(t)\right],\\
		&q_2(t+1)=R_2(q_1(t),q_3(t),q_4(t),q_5(t)),\\
		&q_3(t+1)=(1-l)q_3(t)+l R_3(q_1(t),q_2(t),q_4(t),q_5(t)),\\
		&q_4(t+1)=\frac{2\,q_4(t)+q_1(t)+q_2(t)+q_3(t)+q_5(t)}{2(1+c(q_1(t)+q_2(t)+q_3(t)+q_4(t)+q_5(t))^2)},\\
		&q_5(t+1)=R_5(q_1(t+1),q_2(t+1),q_3(t+1),q_4(t+1)).
		\end{split}
	\right.
\end{split}
\end{equation}
Therefore, the equilibria are described by
\begin{equation}
	\left\{\begin{split}
		&k q_1\left(\frac{q_2+q_3}{(q_1+q_2+q_3+q_4+q_5)^2}-2\,cq_1\right)=0,\\
		&q_1+q_3+q_4+q_5-2\,cq_2(q_1+q_2+q_3+q_4+q_5)^2=0,\\
		&q_1+q_2+q_4+q_5-2\,cq_3(q_1+q_2+q_3+q_4+q_5)^2=0,\\
		&q_4-\frac{2\,q_4+q_1+q_2+q_3+q_5}{2(1+c(q_1+q_2+q_3+q_4+q_5)^2)}=0,\\
		&q_1+q_2+q_3+q_4-2\,cq_5(q_1+q_2+q_3+q_4+q_5)^2=0,
	\end{split}
	\right.
\end{equation}
which could be solved by two solutions
\begin{equation}
\begin{split}
E_{GBALR}^1=&~\left(0,\sqrt{\frac{3}{32c}},\sqrt{\frac{3}{32c}},\sqrt{\frac{3}{32c}},\sqrt{\frac{3}{32c}}\right),\\
E_{GBALR}^2=&~\left(\sqrt{\frac{2}{25c}},\sqrt{\frac{2}{25c}},\sqrt{\frac{2}{25c}},\sqrt{\frac{2}{25c}},\sqrt{\frac{2}{25c}}\right).
\end{split}
\end{equation}

For simplicity, we denote the first and the fourth equation in \eqref{eq:gbalr-map} to be 
$$q_1(t+1)=G_1(q_1(t),q_2(t),q_3(t),q_4(t),q_5(t))$$
and 
$$q_4(t+1)=L_4(q_1(t),q_2(t),q_3(t),q_4(t),q_5(t)),$$
respectively. One may find that \eqref{eq:gbalr-map} could be reformulated to the following $4$-dimensional map.
\begin{equation}\label{eq:gbalr-map4}
\begin{split}
	&T_{GBALR}(q_1,q_2,q_3,q_4): \\
	&\left\{\begin{split}
		&q_1(t+1)=G_1(q_1(t),q_2(t),q_3(t),q_4(t),R_5(q_1(t),q_2(t),q_3(t),q_4(t))),\\
		&q_2(t+1)=R_2(q_1(t),q_3(t),q_4(t),R_5(q_1(t),q_2(t),q_3(t),q_4(t))),\\
		&q_3(t+1)=(1-l)q_3(t)+l R_3(q_1(t),q_2(t),q_4(t),R_5(q_1(t),q_2(t),q_3(t),q_4(t))),\\
		&q_4(t+1)=L_4(q_1(t),q_2(t),q_3(t),q_4(t),R_5(q_1(t),q_2(t),q_3(t),q_4(t))).
		\end{split}
	\right.
\end{split}
\end{equation}
Hence, the analysis of the local stability is transformed to the investigation of the Jacobian matrix \eqref{eq:gbalr-map4} of the form

\begin{equation}
	J_{GBALR}=\left[\begin{matrix}
		\frac{\partial q_1(t+1)}{\partial q_1(t)} & \frac{\partial q_1(t+1)}{\partial q_2(t)} & \frac{\partial q_1(t+1)}{\partial q_3(t)} & \frac{\partial q_1(t+1)}{\partial q_4(t)}\\
		\frac{\partial q_2(t+1)}{\partial q_1(t)} & \frac{\partial q_2(t+1)}{\partial q_2(t)} & \frac{\partial q_2(t+1)}{\partial q_3(t)} & \frac{\partial q_2(t+1)}{\partial q_4(t)}\\
		\frac{\partial q_3(t+1)}{\partial q_1(t)} & \frac{\partial q_3(t+1)}{\partial q_2(t)} & \frac{\partial q_3(t+1)}{\partial q_3(t)} & \frac{\partial q_3(t+1)}{\partial q_4(t)}\\
		\frac{\partial q_4(t+1)}{\partial q_1(t)} & \frac{\partial q_4(t+1)}{\partial q_2(t)} & \frac{\partial q_4(t+1)}{\partial q_3(t)} & \frac{\partial q_4(t+1)}{\partial q_4(t)}\\
	\end{matrix}
	\right],
\end{equation}
where
\begin{equation}
\begin{split}
	\frac{\partial q_1(t+1)}{\partial q_i(t)}=&~\frac{\partial G_1}{\partial q_i}+\frac{\partial G_1}{\partial q_5}\frac{\partial R_5}{\partial q_i},~~i=1,2,3,4,\\
	\frac{\partial q_2(t+1)}{\partial q_i(t)}=&~\frac{\partial R_2}{\partial q_i} + l\frac{\partial R_2}{\partial q_5}\frac{\partial R_5}{\partial q_i},~~i=1,3,4,\\
	\frac{\partial q_2(t+1)}{\partial q_2(t)}=&~\frac{\partial R_2}{\partial q_5}\frac{\partial R_5}{\partial q_2},\\
	\frac{\partial q_3(t+1)}{\partial q_i(t)}=&~l\frac{\partial R_3}{\partial q_i} + l\frac{\partial R_3}{\partial q_5}\frac{\partial R_5}{\partial q_i},~~i=1,2,4,\\
	\frac{\partial q_3(t+1)}{\partial q_3(t)}=&~(1-l)+l\frac{\partial R_3}{\partial q_5}\frac{\partial R_5}{\partial q_3},\\
	\frac{\partial q_4(t+1)}{\partial q_i(t)}=&~\frac{\partial L_4}{\partial q_i}+\frac{\partial L_4}{\partial q_5}\frac{\partial R_5}{\partial q_i},~~i=1,2,3,4.\\
\end{split}
\end{equation}

Likewise, we focus on the positive equilibrium $E^2_{GBALR}$, where the Jacobian matrix $J_{GBALR}$ becomes
\begin{equation}
	J_{GBALR}(E^2_{GBALR})=\left[\begin{matrix}
		1-{31\,k\sqrt{2c}}/56 & -3\,k\sqrt{2c}/{56} & -3\,k\sqrt{2c}/{56} & -3\,k\sqrt{2c}/{56}\\
		-{75}/{784} & 9/784 & -{75}/{784} & -{75}/{784}\\
		0 & 0 & 1-25\,l/28 & 0\\
		-{5}/{56} & -{5}/{56} & -{5}/{56} & {13}/{168}\\
	\end{matrix}
	\right].
\end{equation}

According to Corollary \ref{cor:4-dim-stable}, the unique positive equilibrium $E_{GBALR}^2$ is locally stable if and only if the following condition is satisfied.
\begin{equation}\label{eq:inequality-6-5firm}
\begin{split}
	CD_{GBALR}^1>0,~CD_{GBALR}^2>0,~CD_{GBALR}^3<0,\\
	CD_{GBALR}^4<0,~CD_{GBALR}^5<0,~CD_{GBALR}^6>0,
\end{split}
\end{equation}
where
\begin{equation}
\begin{split}
CD_{GBALR}^1=&~
kl\sqrt{25c/2},\\	
CD_{GBALR}^2=&~
(25\,l-56)(5737\,k\sqrt{25c/2}-50860),\\
CD_{GBALR}^3=&~
(3934321875\,k^3l^3-104905111500\,k^3l^2+1172129631120\,k^3l\\
&-1186719653952\,k^3)(\sqrt{25c/2})^3+(-439562531250\,k^2l^3+19054516460000\,k^2l^2\\
&-144796527937600\,k^2l+134072666053760\,k^2)(\sqrt{25c/2})^2+(19706242500000\,kl^3\\
&-579386747450000\,kl^2-1721529608680000\,kl+3133067852544000\,k)\sqrt{25c/2}\\
&-113004562500000\,l^3+1975821995000000\,l^2+12875890524000000\,l\\
&-37485773024000000,\\
CD_{GBALR}^4=&~
(9423\,k^2(\sqrt{25c/2})^2-981050\,k\sqrt{25c/2}-33575000)((3375\,kl^3-89180\,kl^2\\
&-629552\,kl+812224\,k)\sqrt{25c/2}\\
&-22500\,l^3+343000\,l^2+3332000\,l)\\
CD_{GBALR}^5=&~
(225\,kl-252\,k)\sqrt{25c/2}-1500\,l-217840,\\
CD_{GBALR}^6=&~
(225\,kl-252\,k)\sqrt{25c/2}-1500\,l+221200.\\
\end{split}
\end{equation}

By observing Table \ref{tab:gbalr-sample}, we have the following proposition.

\begin{proposition}
The $T_{GBALR}$ model described by \eqref{eq:gbalr-map} has a unique positive equilibrium 
$$\left(\sqrt{\frac{2}{25c}},\sqrt{\frac{2}{25c}},\sqrt{\frac{2}{25c}},\sqrt{\frac{2}{25c}},\sqrt{\frac{2}{25c}}\right),$$
 which is locally stable provided that
\begin{equation*}
k\sqrt{c}<\frac{ 10172\sqrt{2}}{5737}.
\end{equation*}
\end{proposition}

Furthermore, the following result is acquired.

\begin{proposition}
	The stability region of the $T_{GBALR}$ model is strictly larger than that of $T_{GBAL}$.
\end{proposition}
\begin{proof}
It suffices to prove that 
	$$\frac{2\sqrt{6}(226\,l-441)}{512\,l-1017}<\frac{ 10172\sqrt{2}}{5737},$$
which is equivalent to 
$$10172\sqrt{2}(1017-512\,l) - 5737\times 2\sqrt{6}(441-226\,l)>0,$$
which is true by checking at $l=0$ and $l=1$. 
\end{proof}

\begin{table}[htbp]
	\centering 
	\caption{Stability Condition of $T_{GBALR}$ at Selected Sample Points}
	\label{tab:gbalr-sample} 
	 
	\begin{tabular}{|l|c|c|c|}  
		\hline  % 表格的横线
		 sample point of $(k,l,c)$ & $CD_{GBALR}^1>0$ & $CD_{GBALR}^2>0$ & $CD_{GBALR}^3<0$ \\ 
		\hline 
		(453/256, 61/128, 1/2) & true& true& true \\
		\hline 
		(1007/256, 61/128, 1/2) & true& false& true  \\
		\hline
		(7183/256, 61/128, 1/2)& true& false& false\\
		\hline
		(6675/128, 61/128, 1/2)& true& false& true\\
		\hline
		(10587/32, 61/128, 1/2)& true& false& true\\
		\hline
		(9755/16, 61/128, 1/2)& true& false& true\\
		\hline
		(453/256, 251/256, 1/2)& true& true& true\\
		\hline
		(1567/256, 251/256, 1/2)& true& false& true\\
		\hline
		(225/8, 251/256, 1/2)& true& false& false\\
		\hline
		(12807/256, 251/256, 1/2)& true& false& true\\
		\hline
		(91267/64, 251/256, 1/2)& true& false& true\\
		\hline
		(89603/32, 251/256, 1/2)& true& false& true\\
		\hline
	\end{tabular}
	
	\begin{tabular}{|l|c|c|c|}  
		\hline  % 表格的横线
		 sample point of $(k,l,c)$ & $CD_{GBALR}^4<0$ & $CD_{GBALR}^5<0$ & $CD_{GBALR}^6>0$ \\ 
		\hline 
		(453/256, 61/128, 1/2) & true& true& true \\
		\hline 
		(1007/256, 61/128, 1/2) & true& true& true  \\
		\hline
		(7183/256, 61/128, 1/2)& true& true& true\\
		\hline
		(6675/128, 61/128, 1/2)& true& true& true\\
		\hline
		(10587/32, 61/128, 1/2)& false& true& true\\
		\hline
		(9755/16, 61/128, 1/2)& false& true& false\\
		\hline
		(453/256, 251/256, 1/2)& true& true& true\\
		\hline
		(1567/256, 251/256, 1/2)& true& true& true\\
		\hline
		(225/8, 251/256, 1/2)& true& true& true\\
		\hline
		(12807/256, 251/256, 1/2)& true& false& true\\
		\hline
		(91267/64, 251/256, 1/2)& false& true& true\\
		\hline
		(89603/32, 251/256, 1/2)& false& true& false\\
		\hline
	\end{tabular}

\end{table}

\section{Concluding Remarks}

\begin{figure}[htbp]
    \centering
    \includegraphics[width=7cm]{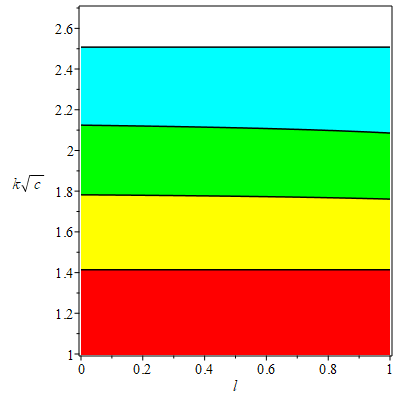}
    \caption{The stability regions of the models considered in the paper. The unique equilibrium of $T_{GB}$ is locally stable if and only if the parameters take values from the red region. The unique positive equilibrium of $T_{GBA}$ is locally stable if and only if the parameters take values from the red and yellow regions. By analogy, similar conclusions can be obtained for the $T_{GBAL}$ and $T_{GBALR}$ models.}
    \label{fig:regions}
\end{figure}

%\section*{Acknowledgments}
%
%The author is grateful to the anonymous referees for their helpful comments. This work has been supported by Philosophy and Social Science Foundation of Guangdong (No. GD21CLJ01), Major Research and Cultivation Project of Dongguan City College (No. 2021YZDYB04Z) and Social Development Science and Technology Project of Dongguan (No. 20211800900692). 

\bibliographystyle{abbrv}
\bibliography{ref_oligopoly}

%\section*{Appendix}
%\tiny
%\begin{align*}
%
%\end{autobreak}\\
%
%\begin{autobreak}
%
%\end{align*}

\end{CJK}
\end{document}

%% file: macro.tex
\usepackage{amssymb}
\usepackage{amsmath}
\usepackage{amsthm}
\usepackage{amsfonts}
\usepackage{mathtools}
\usepackage{mathrsfs}
\usepackage{bm} % 处理数学公式中的黑斜体的宏包

\usepackage{graphicx}
\usepackage{dcolumn}
\usepackage[ruled,lined,algonl]{algorithm2e}
\usepackage{color}
\usepackage{autobreak}
\allowdisplaybreaks

\theoremstyle{plain}

\newtheorem{corollary}{Corollary}

\newtheorem{proposition}{Proposition}

\theoremstyle{definition}

\theoremstyle{remark}